\documentclass[12pt]{amsart}
\usepackage{geometry}
\usepackage{amssymb}
\usepackage[foot]{amsaddr}
\usepackage{graphicx}
\usepackage{xcolor}
\usepackage{enumitem}
\usepackage{amsthm}
\usepackage{hyperref}
\hypersetup{
pdftitle={},%
pdfauthor={},%
pdfsubject={},%
pdfkeywords={},%
colorlinks=true,%
linkcolor=blue,%
citecolor=red,%
linktocpage=true,%
pageanchor=true
}
\newtheorem{theorem}{Theorem}[section]
\newtheorem{corollary}{Corollary}[theorem]
\newtheorem{proposition}{Proposition}[section]
\newcommand{\eqn}[1]{(\ref{#1})}
\newcommand{\C}{\mathbb{C}}
\newcommand{\Tr}{\mathrm{Tr}}
\newcommand{\be}{\begin{equation}}
\newcommand{\ee}{\end{equation}}
\def\beqa{\begin{eqnarray}}
\def\eeqa{\end{eqnarray}}
\def\bean{\begin{eqnarray*}}
\def\eean{\end{eqnarray*}}
\def\nn{\nonumber}
\usepackage{cite}

\begin{document}
\author{{\large M. A. Man'ko,$^1$ V. I. Man'ko,$^{1,2}$ G.~Marmo,$^{3,4}$ F.
Ventriglia,$^{3,4}$ and P. Vitale~$^{3,4}$}\\
%\address
$^1$ Lebedev Physical Institute, Leninskii Prospect 53,
Moscow 119991, Russia \\
$^2$ Moscow Institute of Physics and
Technology (State University), Institutskii per. 9, Dolgoprudny, Moscow Region
141700, Russia\\
$^3$ Dipartimento di Fisica ``E. Pancini,'' Universit\`a di Napoli Federico II,
Complesso Universitario di Monte S. Angelo Edificio 6, via Cintia, 80126
Napoli, Italy\\
$^4$ INFN-Sezione di Napoli, Complesso Universitario di Monte S. Angelo, Edificio 6,
via Cintia, 80126 Napoli, Italy}

\title {\large Dichotomic probability representation of quantum states}

\let\origmaketitle\maketitle
\def\maketitle{
  \begingroup
  \def\uppercasenonmath##1{} % this disables uppercasing title
  \let\MakeUppercase\relax % this disables uppercasing authors
  \origmaketitle
  \endgroup
}
\email{mankoma@lebedev.ru}\email{manko@lebedev.ru}\email{  marmo@na.infn.it}\email{
ventriglia@na.infn.it}\email{ patrizia.vitale@unina.it}
\maketitle
\begin{abstract}
We present systematic proofs of statements about probability
representations of qudit density states in terms of standard probability distributions of dichotomic random variables.  New
relations and new entropic-information inequalities are derived. The examples of
3- and 4- level states  are explicitly worked out. 
% $\hat{\mathbb{I}}$ $\mathbb{I}$ ${\mathbf{1}}$
%\newcommand{\Id}{\mathds{1}}
\end{abstract}

\vspace{2pc}


% $\mathbb{I}$

%\pacs{42.50.Ct, 03.65.Fd, 64.70.Tg}

%\submitto{\JPA}
%\maketitle
%\ioptwocol

\section*{Introduction}

Quantum states are associated with rays of a Hilbert space,  or, in general, 
with density operators acting on
it~\cite{Landau27,vonNeumann27,Dirac-book,Schroed1926,Wigner32,Husimi40,Kano65,Glauber63,Sudar63}.
In either case, quantum states do  replace  classical probability distributions but they do not 
represent  fair probabilities on any  sample space. However, they  can be  associated with true probability distributions   in the tomographic picture of quantum mechanics. There, 
states are identified with tomographic-probability
distributions (quantum tomograms) of homodyne quadrature components
for continuous variables or spin tomographic-probability
distributions for discrete spin-projection variables, see the discussion e.g. 
in~\cite{TombesiPLA,DodPLA,OlgaJETP,MarmoPS150,Vitale16,scully}.

Optical tomograms for systems with continuous variables were measured in
 experiments~\cite{Raymer,Raymer-Lvovski}, where the Wigner function of
photon states was reconstructed, by means of the relation between  tomograms
and Wigner functions through the Radon transform~\cite{Radon1917} found
in~\cite{BerBer,VogRis}. 

In~\cite{TombesiPLA},  symplectic-tomography for 
 photon quadrature was introduced, as an alternative to   optical tomography. Also in~\cite{TombesiPLA},
tomographic-probability distributions were suggested as a
primary notion of photon states, providing a  ``classical-like'' description of quantum
states in a statistical mechanics approach. The tomographic characterization  of spin states was introduced 
in~\cite{DodPLA,OlgaJETP}; see also the review~\cite{MarmoPS150}.

Recently
%\cite{ChernegaJRLR-2-2017,ChernegaJRLR-4-2017,ChernegaJRLR-6-2017,VovaJPCS-Bregenz2018,ChernegaJRLR-2018,
%RitaEntropy,PS-Milestones,Julio-Entropy,MA-Entropy}
the possibility to parameterise  density matrices of qudit states
($d$-level atom states, spin states) by sets of dichotomic
probability distributions has been proposed and developed
in~\cite{Chernega-JRLR,Chernega-JPCS,RitaEntropy,PS-Milestones,Julio-Entropy,MA-Entropy}.
The aim of this work is to provide general statements about the
probability description of qudit states by means of dichotomic
probabilities and  to prove new properties of nonnegative
trace-one Hermitian matrices. 

The approach is quite simple, out of the $d$-dimensional Hilbert space  $\C^d$, we identify $d(d-1)$  complex planes, in each one of them we consider a group of unitary matrices isomorphic with $U(2)$, and matrices  corresponding to the associated $\mathfrak{u}(2)$ Lie algebra. Clearly this family of $\mathfrak{u}(2)$  Lie algebras is sufficient to describe the whole $\mathfrak{u}(d)$ algebra, in some redundant way. In some sense we could say that a generic quantum system may be studied by means of properly chosen qubit systems (this should not be confused with the description of the total system by means of  qubit subsystems.  Indeed, in the latter situation,  the associated composite system would be a tensor product of qubits, which is not the case with our generic Hilbert space). We then use the tomographic description of each $\mathfrak{u}(2)$ subalgebra: because of dimensionality, the tomograms will represent dichotomic probability distributions.  We shall prove that an arbitrary
$d$$\times$$d$ density matrix can be parameterized by $(d^2-1)$
probability distributions of dichotomic random variables. Having elaborated these tools, we use them to decompose a given $d$-state, with $d=nm$,  into $n$-states and $m$-states. Using
this result, we obtain new entropic-information inequalities among
matrix elements of an arbitrary matrix $\rho$, associated with a state. In addition, we
obtain new inequalities for characteristic polynomials associated
with such matrices. We illustrate some of these claims in the case
of qubit and qutrit density matrices.  This proposal will be worked out in details in the coming sections. We shall use a pedagogical style and spell out all computations.

The paper is organized as follows.

In section~\ref{qspb}, we exhibit the dichotomic probability
representation of qubit and qutrit states. In section~\ref{qNit}, we
generalize the dichotomic probability representation to qudit
states. In section~\ref{redumats}, we illustrate a reduction
procedure to construct two types of  new density states from a
$d=nm$-dimensional starting one, with dimension $n$ and $m$,
respectively. In section~\ref{infoentro}, we obtain some new
entropic inequalities for the  matrix elements of density states, as
spin~off of the approach developed. 
%In section~\ref{diss}, we
%present an example of the qubit state parameterized by the
%probabilities associated with mixed states $\rho_\alpha^j$. 
Finally,
we give our conclusions in section~\ref{concl}.

%\newpage

\section{Quantum states and probability vectors}\label{qspb}

Quantum states and probability vectors of dichotomic observables can be
considered within the probability representation of quantum mechanics, where
the states are usually considered to replace fair probabilities. 

For simplicity, we restrict our
considerations to finite-dimensional Hilbert spaces. In what follows we shall take a pedagogical attitude and spell out all details so that general statements are always illustrated by the example.

Let ${\mathcal H}$ be the Hilbert space associated with our quantum system. If
$\mathcal{H}$ and $ |  e_1\rangle, |  e_2\rangle, \ldots, |
e_d\rangle$ is an orthonormal basis, we can associate a matrix with
$ | \psi\rangle$, by first defining a rank-one projector and then using a specific basis of orthonormal vectors:
\begin{equation}\label{1}
\|\psi_{jk}^{(e)}\|=\Big{\|}\frac{\langle e_j | \psi\rangle\langle
\psi | e_k\rangle}{\langle\psi | \psi\rangle}\Big{\|}.
\end{equation}
The diagonal elements $\{\psi_{jj}^{(e)}\}$ represent a probability
distribution with $n$ components; $\psi_{jj}^{(e)}\geq 0$ and
$\sum_j\psi_{jj}^{(e)}=1$. It is a probability distribution on the set
$\{1,2,\ldots,d\}$, we call it a probability vector.

If we select a different orthonormal basis, say, $ |
f_1\rangle, |  f_2\rangle, \ldots, |  f_n\rangle$, we associate
with $ | \psi\rangle$ a different matrix
\begin{equation}\label{2}
\|\psi_{jk}^{(f)}\|=\Big{\|}\frac{\langle f_j | \psi\rangle\langle
\psi | f_k \rangle}{\langle\psi | \psi\rangle}\Big{\|}.
\end{equation}
Again, the diagonal elements provide a new probability vector, a new
probability distribution on the set $\{1,2,\ldots,d\}$.

Thus, with a given vector, depending on the chosen basis, we associate a
family of probability distributions on the set $\{1,2,\ldots,d\}$. As a matter
of fact, with a given $ | \psi\rangle$ but using different orthonormal bases,
one can obtain all possible probability distributions. As a matter of fact,
given a probability vector, say $(p_1,p_2,\ldots,p_d)$, it is possible to find
a whole family of states corresponding to the same probability distribution;
indeed, in the given basis we would have
$ | \psi\rangle\langle \psi|=\sum_{j,k}\sqrt{p_j p_k}e^{i(\varphi_j-\varphi_k)} |  e_j\rangle\langle e_k|$, where
$\varphi_j$ is completely arbitrary, by using different bases we would get
different states $ | \psi_f\rangle\langle \psi_f |=\sum_{j,k}\sqrt{p_j p_k}e^{i(\varphi_j-\varphi_k)} |  f_j\rangle\langle f_k|$.

We notice that states we build are rank-one Hermitian operators of trace
one.

The identification of the probability distribution with the diagonal of the
matrix associated with a given vector $ | \psi\rangle$ shows also that the
association disregards all ``off-diagonal'' elements, i.e., not only rank-one
operators but also states of  higher rank, as long as the diagonal is
unchanged, will give rise to the same probability distribution.

The probability distributions associated with every vector by means of
different bases are called tomograms; indeed, tomography consists of
reconstructing the state when a sufficient set (a ``quorum'') of tomograms is
provided. Such tomograms of spin states were studied, e.g.,
in~\cite{Weigert1,Weigert2}. The spin tomography was discussed
in~\cite{TombesiPLA,DodPLA}; see also the review~\cite{Paini}.

Having stressed that alternative states, both pure and mixed, may give rise to
the same probability distribution, it is quite surprising and highly not
trivial that by giving a sufficient set of probability distributions thought
of as related to the same state, we are able to reconstruct uniquely the
state, be it pure or mixed.

Let us identify the mathematical ingredients of previous construction. We have
first the association of a rank-one projector with every vector, say,
$ | \psi\rangle\to\displaystyle{\frac{ | \psi\rangle\langle
\psi | }{\langle\psi | \psi\rangle}}=\rho_\psi$. Next, the selection of a
basis in ${\mathcal H}$ provides a resolution of the identity, say,
${\mathbf{1}}=\sum_j |  e_j\rangle\langle e_j | $ and, moreover, allows for
the construction of a basis of Hermitian matrices, specifically,
$$
 |  e_j\rangle\langle e_j | ,\quad  |  e_j\rangle\langle
e_k |  + |  e_k\rangle\langle e_j | ,\quad i( |
e_j\rangle\langle e_k |  - |  e_k\rangle\langle e_j | ),
\quad
 |  e_k\rangle\langle e_k | 
$$
when j and k  go from 1 to n.
However, when we restrict to  given pair of j and k, 
 in each ``$(j,k)$-plane'' we build a basis of a  $u(2)$ Lie
algebra. Thus we obtain a different "placement" of an abstract u(2) Lie algebra, for any choice of a $(j,k)$ plane. The unitary transformation taking from one basis to a different one changes the placement of the u(2) Lie algebra.

The Weyl basis $\left\{ |  e_j\rangle\langle e_k | \right\}$
allows for the construction of the matrix associated with any vector
$ | \psi\rangle$; we have
$\psi_{jk}=\mbox{Tr}\left(\displaystyle{\frac{ | \psi\rangle\langle\psi | }{\langle\psi | \psi\rangle}\,
 |  e_j\rangle\langle e_k | }\right)$.

The association of a probability distribution with $\rho_\psi$ only
uses diagonal elements $\left\{ |  e_j\rangle\langle
e_j | \right\}$; thus, we need a sufficient number of independent
bases so that by means of the diagonal elements associated with the
various bases, say, $\big[\left\{ |  e_j\rangle\langle
e_j | \right\},\left\{ |  f_j\rangle\langle
f_j, | \right\},\ldots,\left\{ |  k_j\rangle\langle
k_j, | \right\},\ldots\big]$, we may generate a basis of rank-one
operators.

In order to fully reconstruct a state we need $d^2-1= (d-1)(d+1)$ parameters. On  using resolutions of the identity
$$\sum_j |  f_j\rangle\langle
f_j | =\sum_j |  e_j\rangle\langle e_j | =\cdots=\sum_j |
k_j\rangle\langle k_j | ={\bf 1},$$ the independent diagonal elements
associated with every basis are $(d-1)$ in number; therefore  we need
$(d+1)$ of such independent families.

Since orthonormal bases may be constructed by means of normalized
eigenvectors of a generic observable $A$ with simple  eigenvalues, to obtain full information on the quantum state, we can
measure $(d+1)$ independent families of $(d-1)$ pairwise commuting
observables, which are independent. From each family, it would be enough to measure  just one observable which has a non degenerate spectrum.

\noindent {\bf{Remark:}} By using the expectation value functions,
$e_A(\psi)=\displaystyle{\frac{\langle\psi |  A | \psi\rangle}
{\langle\psi | \psi\rangle}}$, we may define the independence to be
the functional independence of the expectation value functions
associated with every observable of the pairwise commuting family.

To nail down these general considerations, we consider two examples,
namely, a qubit and a qutrit.

\subsection {The qubit case}

Here, $d=2$ and  ${\mathcal H}={\C}^2$. We have to measure $d+1=3$ independent families of $d-1=1$ commuting observables, which we choose to be one of  the Pauli matrices
$\sigma_1$, $\sigma_2$, and $\sigma_3$, which are  Hermitian operators in
the space of 2$\times$2 matrices. Clearly, each Pauli matrix, together with the identity matrix $\sigma_0$ will define a
basis of $\mathfrak{u}_2$.

For the observable associated to $\sigma_3=\left(\begin{array}{cc}1&0\\
0&-1\end{array}\right)$ we have
\begin{equation}
 | f^+_3\rangle=\left(\begin{array}{c}1\\0\end{array}\right),\;\;\;\;
 | f^-_3\rangle=\left(\begin{array}{c}0\\1\end{array}\right).
\end{equation}
 The two vectors $ |
f^+_3\rangle=\left(\begin{array}{cc}1\\0\end{array}\right)$ and $ |
f^-_3\rangle=\left(\begin{array}{cc}0\\1\end{array}\right)$, being
orthonormal eigen-vectors for $\sigma_3$,
 determine rank-one projectors
$\hat {\Pi}_3^+, \hat {\Pi}_3^-$ %k=1,2
\begin{equation}\label{A1}
\hat {\Pi}_3^+= |  f^+_3\rangle\langle {f^+_3} | =\left(\begin{array}{cc}1&0\\
0&0\end{array}\right),\quad
\hat {\Pi}_3^-= |  f^-_3\rangle\langle f^-_3 | =\left(\begin{array}{cc}0&0\\
0&1\end{array}\right)
\end{equation}
besides  the matrices
\begin{equation}\label{A2}
\hat{\Pi}^\pm_3= |  {f^+}_3\rangle\langle {f^-}_3 | =\left(\begin{array}{cc}0&1\\
0&0\end{array}\right),\quad
\hat{\Pi}^\mp_3= |  {f^-}_3\rangle\langle {f^+}_3 | =\left(\begin{array}{cc}0&0\\
1&0\end{array}\right).
\end{equation}
Altogether they  form a basis in the linear space of
2$\times$2-matrices.

For the observable associated to $\sigma_1=\left(\begin{array}{cc}0&1\\
1&0\end{array}\right)$ the orthonormal eigenvectors are
\begin{equation}\label{A3}
 |  f^+_1\rangle =\frac{1}{\sqrt 2}\left(\begin{array}{c}1\\
1\end{array}\right),\quad |  f^-_1\rangle=\frac{1}{\sqrt 2}\left(\begin{array}{c}1\\
-1\end{array}\right),
\end{equation}
They  determine the rank-one projectors $\hat {\Pi}_1^+, \hat {\Pi}_1^-$
\begin{equation}\label{A4}
 \hat {\Pi}_1^+= |  f^+_1\rangle\langle f^+_1 | =\left(\begin{array}{cc}1/2&1/2\\
1/2&1/2\end{array}\right),\quad
\hat {\Pi}_1^-= |  f^-_1\rangle\langle f^-_1 | =\left(\begin{array}{cc}1/2&-1/2\\
-1/2&1/2\end{array}\right)
\end{equation}
and the matrices
\begin{equation}\label{A5}
   \hat{\Pi}^\pm_{1}= |  {f^+}_1\rangle\langle {f^-}_1 | =\left(\begin{array}{cc}1/2&-1/2\\
1/2&-1/2\end{array}\right),\quad
\hat{\Pi}^\mp_{1}= |  {f^-}_1\rangle\langle {f^+}_1 | =\left(\begin{array}{cc}1/2&1/2\\
-1/2&-1/2\end{array}\right)
\end{equation}
yielding  another  basis in the linear space of
2$\times$2-matrices.

Finally, for the observable associated to the Pauli matrix $\sigma_2=\left(\begin{array}{cc}0&-i\\
i&0\end{array}\right)$, the orthonormal basis of eigenvectors is
\begin{equation}\label{A6}
 |  f^+_2\rangle=\frac{1}{\sqrt 2}\left(\begin{array}{cc}1\\i\end{array}\right),\quad
 |  f^-_2\rangle=\frac{1}{\sqrt 2}\left(\begin{array}{cc}1\\
-i\end{array}\right).
\end{equation}
They  determine the rank-one
projectors $\hat {\Pi}_2^+, \hat {\Pi}_2^-;$
\begin{equation}\label{A7}
  \hat{\Pi}^+_{2}= |  f^+_2\rangle\langle f^+_2 | =\left(\begin{array}{cc}1/2&-i/2\\
i/2&1/2\end{array}\right),\quad
\hat{\Pi}^-_{2}= |  f^-_2\rangle\langle f^-_2 | =\left(\begin{array}{cc}1/2&i/2\\
-i/2&1/2\end{array}\right)
\end{equation}
and the matrices
\begin{equation}\label{A8}
  \hat{\Pi}^\pm_{2}= |  f^+_2\rangle\langle f^-_2 | =\left(\begin{array}{cc}1/2&i/2\\
i/2&-1/2\end{array}\right),\quad
\hat{\Pi}^\mp_{2}= |  f^-_2\rangle\langle f^+_2 | =\left(\begin{array}{cc}1/2&-i/2\\
-i/2&-1/2\end{array}\right).
\end{equation}
which  form yet another basis in the linear space of
2$\times$2-matrices.
We introduce for future convenience  a compact notation for all  rank one-projectors, namely
\be\label{proj2}
\hat{\Pi}^+_{a}= \frac{1}{2}(\sigma_0+ \sigma_a); \;\;\; \hat{\Pi}^-_{a}= \frac{1}{2}(\sigma_0- \sigma_a).
\ee
For each one of these bases, by  using just rank-one projectors over positive eigenstates, we may associate  two-dimensional probability
vectors, say, $(p_1,1-p_1)$; $(p_2,1-p_2)$; $(p_3,1-p_3)$ to a given state $\rho$, be it pure or mixed. We have indeed
\begin{equation}\label{probvec}
  p_1(\rho)=\Tr \rho \hat{\Pi}^+_{1},  \quad
 p_2(\rho)=\Tr \rho \hat{\Pi}^+_{2},  \quad
 p_3(\rho)=\Tr \rho \hat{\Pi}^+_{3},  \quad
%frac{\langle \psi  |  \hat{\Pi}^+_{1} |
%\psi \rangle}{\langle\psi | \psi\rangle}\,,\quad
%p_2(\psi)=\frac{\langle \psi  |  \hat{\Pi}^+_{2} |
%\psi \rangle}{\langle\psi | \psi\rangle}\,,\quad
%p_3(\psi)=\frac{\langle \psi  |  \hat{\Pi}^+_{3} |
%\psi \rangle}{\langle\psi | \psi\rangle}\,.
\end{equation}
and  analogous relations for $1-p_i$ in terms of rank-one projectors of negative eigenstates. Eqs. \eqn{probvec} define genuine probabilities $0\le p_i\le 1$ because $\rho$ is a positive Hermitean matrix while $\hat{\Pi}^+_{a}$ are rank-one projectors.

In order to discuss the dependence of the dichotomic probability representation on the choice of rank-one projectors, let us observe that  Eq. \eqn{proj2} for $\hat\Pi^+_a$  may be rewritten in the following form
\be
\hat{\Pi}^+_{a}= \frac{1}{2}(\sigma_0+ \vec x_a \cdot \vec \sigma),\;\;\;\;\;\;  (\vec x_a)_j= \delta_{aj}, a,j=1,2,3.
\ee
This implies that, upon rotation of the three vectors $\vec x_a$, we obtain  rotated projectors $(\hat{\Pi}^+_a)'$.
By means of the standard double covering of $SO(3)$ by $SU(2)$, we have indeed
\be\label{xrot}
\vec x'_a= R \vec x_a \longrightarrow  (\vec x_a \cdot \vec \sigma)'= U(\vec x_a \cdot \vec \sigma) U^\dag, \;\; R\in SO(3), U\in SU(2)
\ee
so that
\be\label{Pirot}
(\hat{\Pi}^+_{a})'= U \hat{\Pi}^+_a U^\dag
\ee
and
\be\label{rotlaw}
p'_a(\rho)=\Tr \rho (\hat{\Pi}^+_{a})' = \Tr U^\dag \rho U \hat{\Pi}^+_{a}
\ee
yielding the transformation law  of the  dichotomic probabilities under rotation of rank-one projectors. This result can be easily generalized to the $d$-dimensional case, as we shall see in next section.

By means of these dichotomic probabilities it can be easily checked by direct computation that we can reconstruct the
 state by setting
\begin{equation}\label{3}
 \rho=(\sigma_0/2)+(p_1-1/2)\sigma_1+(p_2-1/2)\sigma_2+(p_3-1/2)\sigma_3.
\end{equation}
We notice, although trivial for $d=2$, that the latter is equivalent to the tomographic approach, where, given a state
\begin{equation}
\rho= \frac{1}{2}\left(\sigma_0+ \vec y \cdot \vec \sigma \right)
\end{equation}
%we can associate to it dichotomic probability vectors $(p_j, 1-p_j), j=1,..,3$,
we have the tomographic relation (see for example \cite{Vitale16})
\begin{equation}
y_j= 2(\mathcal{W}_j-\frac{1}{2})
\end{equation}
with  $\mathcal{W}_j=p_j, j=1,..,3,$ the tomographic  probabilities.

Conversely, given a family of dichotomic probabilities $(p_j, 1-p_j), j=1,..,3,$ Eq. \eqn{3}  can be chosen as a definition of a mixed state $\rho$. Indeed, the latter is Hermitean and can be checked to verify $\Tr \rho =1$. Moreover, its determinant is nonnegative if   the
coefficients  satisfy the inequality
\begin{equation}\label{4}
(p_1-1/2)^2+(p_2-1/2)^2+(p_3-1/2)^2\leq 1/4
\end{equation}
If the dichotomic variables are not correlated, we
have
\begin{equation}\label{5}
  (p_1-1/2)^2+(p_2-1/2)^2+(p_3-1/2)^2\leq 3/4,
\end{equation}
the only constraint being $0\leq p_j\leq
1,\quad j=1,2,3$.
Notice that the inequalities are also satisfied if we use $(1-p_j)$
instead of $p_j$.

Eq. \eqn{3} describes a mixed state in terms of tomograms. If the inequality is saturated, we are dealing with pure states,
i.e., $\rho^2=\rho$.

%Conversely, given a mixed state, which can be represented as
%\begin{equation}
%\rho= \frac{1}{2}\left(\sigma_0+ y_j \sigma_j\right)
%\end{equation}
%we can associate to it dichotomic probability vectors $(p_j, 1-p_j), j=1,..,3$,
%it being (see \cite{Vitale16})
%\begin{equation}
%y_i= 2(p_i-\frac{1}{2}).
%\end{equation}
 Thus, out of three dichotomic probability
distributions, we have been able to reconstruct a state.

Finally, in order to make contact with the coming sections and make it easier to generalize the results to higher dimensions, it is useful to rewrite Eq. \eqn{3} in the Weyl basis for $\mathfrak{u}_2$. To this, let us introduce an orthonormal basis in ${\C}^2$, say $\{ |  e_1\rangle,  |  e_2\rangle\}$. The Weyl basis is represented by four rank-one operators
\begin{equation}
 E_{jk}=  |  e_j\rangle\langle e_k | ,\quad j,k=1,2 .
\end{equation}
By expressing the Hermitian $\mathfrak{u}_2$ generators in terms of the latter
\begin{eqnarray}
\frac{1}{2}\sigma_0&=& \frac{1}{2}(E_{11}+E_{22}), \;\; \frac{1}{2}\sigma_3= \frac{1}{2}(E_{11}-E_{22}) \nonumber\\
\frac{1}{2}\sigma_1&=& \frac{1}{2}(E_{12}+E_{21}), \;\; \frac{1}{2}\sigma_2= \frac{i}{2}(E_{12}-E_{21})
\end{eqnarray}
Eq. \eqn{3} may be rewritten according to
\begin{eqnarray}
\rho&=& E_{11} p_3 +E_{22} (1-p_3) \nonumber\\
&+&  E_{12}\left[(p_1-\frac{1}{2}) - i (p_2-\frac{1}{2})\right]+
E_{21}\left[(p_1-\frac{1}{2}) + i (p_2-\frac{1}{2})\right].
\end{eqnarray}
We notice, for future convenience, that  the  diagonal elements are directly expressed in terms of  the diagonal rank-one projectors associated to $\sigma_3$,  it being
\be
E_{11} = \hat {\Pi}_3^+; \;\;\; E_{22} = {\hat {\Pi}_3^-}.
\ee
Moreover the pertinent matrix entries $\rho_{jj}$ are completely determined by either the positive-eigenvalue projector or the negative one, it being
\be
p_3= \Tr (\rho\hat{\Pi}^+_{3}); \;\;\;  1-p_3= \Tr (\rho\hat{\Pi}^-_{3}).
\ee
Therefore we can rewrite the state $\rho$ as follows
\be
\rho= p_3 E_{11}  + (1-p_3)E_{22}  + \sum_{j \ne k}\rho^{(jk)} E_{jk}
\ee
with
\be
\rho^{(jk)}= \Tr \rho E^T_{jk}.
\end{equation}
This remark will be relevant for higher level systems. This example shows very well how pure states may give rise to all probability vectors in the "classical simplex",indeed the terms with j not equal to k play no role in the expression of the probability vector and their role is simply to change the rank of the state which is being represented.

\subsection{The  qutrit case}
%To define the qutrit state uniquely, we may decide to use the smallest one among $p$ and $(1-p)$.
According to  our previous considerations and notation, here $d=3$,  the Hilbert
space is ${\mathcal H}\equiv { \mathbb C}^3$   and we have to measure $d+1= 4$ independent families of $d-1=2$ commuting observables, which, as for the two-levels system,  can be chosen to be a Cartan subalgebra of the relevant Lie algebra, here  $\mathfrak{u}(3)$, in four different realizations. For each choice of Cartan subalgebra, to which we add the identity, their joint diagonalization yields  three  eigenvectors, which play the role of the eigenvectors of Pauli matrices in the previous subsection. This is the tomographic approach, which allows to  reconstruct the state as it is  detailed  in \cite{Vitale16}.

Such a procedure becomes however difficult to apply in practice with increasing number of levels.  The approach we want to pursue in this paper is, instead,
to  use our knowledge of the two-levels system and characterize all parameters of a qudit state in terms of dichotomic probabilities which are amenable  to the $\mathfrak{u}(2)$ subalgebras of the relevant $\mathfrak{u}(n)$.

To this,  let us consider an   orthonormal basis in  $ { \mathbb C}^3$, say $ |  e_1\rangle$,  $ |
e_2\rangle$,  $ |  e_3\rangle$ and let us  construct for $\mathfrak{u}(3)$ the Weyl  basis
of nine rank-one operators, $E_{jk}= |  e_j\rangle\langle e_k | , j,k=1,..,3$
\begin{eqnarray}\label{u3}
 |  e_1\rangle\langle e_1 | \quad  |  e_1\rangle\langle
e_2 | \quad  |  e_1\rangle\langle e_3 | \nonumber\\
 |  e_2\rangle\langle e_1 | \quad  |  e_2\rangle\langle
e_2 | \quad  |  e_2\rangle\langle e_3 | \\
 |  e_3\rangle\langle e_1 | \quad  |  e_3\rangle\langle
e_2 | \quad  |  e_3\rangle\langle e_3 | \nonumber
\end{eqnarray}
all of them providing a representation of a pair
groupoid~\cite{Gagloi,Swinger}. Let us notice that an alternative basis for $\mathfrak{u}(3)$ is represented by the eight Gell--Mann matrices $\lambda_i$ to which we add  the identity. The latter, which was used in  \cite{Vitale16}, is however not convenient for the present purposes, and, once again, not immediately generalizable to higher dimensions.

It is now easy to see that, in a natural way, we have
the possibility to define three different $\mathfrak{u}(2)$ bases, namely,
\begin{enumerate}
\item[] ${\mathfrak{u}(2)}^{12} :\quad  \left\{  |  e_1\rangle\langle e_1 | ,\quad  |  e_1\rangle\langle
e_2 | ,\quad  |  e_2\rangle\langle e_1 | ,\quad  |  e_2\rangle\langle e_2 |  \right\}$
\item[] ${\mathfrak{u}(2)}^{13}: \quad \left\{ |  e_1\rangle\langle e_1 | ,\quad  |  e_1\rangle\langle
e_3 | ,\quad  |  e_3\rangle\langle e_1 | ,\quad  |
e_3\rangle\langle e_3 |  \right\}$
\item[] ${\mathfrak{u}(2)}^{23}:\quad \left\{  |  e_2\rangle\langle e_2 | ,\quad  |  e_2\rangle\langle
e_3 | ,\quad  |  e_3\rangle\langle e_2 | ,\quad  |
e_3\rangle\langle e_3 | \right\}$
\end{enumerate}
%\end{eqnarray}
which are obtained by the array \eqn{u3} removing, in the order, the third row and third column, the second row and the second column, the first row and the first column.

For each $\mathfrak{u}(2)$, namely, for each choice of   $(jk)$,  $j,k\in(1,2,3)$ and $j<k$,  we can realize  Hermitean $\mathfrak{u}_2$ generators $S_\mu, \mu= 0,..,3$ acting on the $(jk)$ plane, according to
\begin{eqnarray}
S_0^{(jk)}=\frac{1}{2}(E_{jj}+E_{kk}) \quad  S_3^{(jk)}= \frac{1}{2}(E_{jj}-E_{kk})\\
S_1^{(jk)}=\frac{1}{2}(E_{jk}+E_{kj}) \quad  S_2^{(jk)}=- \frac{i}{2}(E_{jk}-E_{kj})\,.
\end{eqnarray}
Since
\be
\Tr E^T_{jk} E_{mn}= \delta_{km}\delta_{jn}
\ee
the Hermitean $\mathfrak{u}(2) $ generators $S_\mu^{(jk)}$ are orthonormal with respect to the scalar product $\langle A |  B\rangle=\mbox{Tr}\,A^\dagger B$.

For each $\mathfrak{u}(2)$ we can apply the procedure described in previous section to obtain rank-one projectors. We consider the  eigenvector $|f^+\rangle$ of positive eigenvalue, for each Hermitean generator of each $\mathfrak{u}(2)$ algebra, namely: $|f^{(12)}_{1}\rangle, |f^{(12)}_{2}\rangle,|f^{(12)}_{3}\rangle$, $|f^{(13)}_{1}\rangle,|f^{(13)}_{2}\rangle,|f^{(13)}_{3}\rangle$, $|f^{(23)}_{1}\rangle,|f^{(23)}_{2}\rangle,|f^{(23)}_{3}\rangle$ (where we have omitted the superscript $+$; we could have chosen to work with the eigenvectors of negative eigenvalue, as shown in the previous section) and we construct rank one-projectors
\be\label{proj}
\hat{\Pi}^{(jk)}_a= |f^{(jk)}_a\rangle\langle f^{(jk)}_a|= S_0^{(jk)} + S_a^{(jk)},\;\;\; a= 1,..,3
\ee
which explicitly read
\begin{eqnarray*}
  \hat{\Pi}^{(12)}_1=\left(\begin{array}{ccc}
1/2&1/2&0\\
1/2&1/2&0\\
0&0&0\end{array}\right),\quad \hat{\Pi}^{(12)}_2=\left(\begin{array}{ccc}
1/2&-i/2&0\\
i/2&1/2&0\\
0&0&0\end{array}\right),\quad \hat{\Pi}^{(12)}_3=\left(\begin{array}{ccc}
1&0&0\\
0&0&0\\
0&0&0\end{array}\right),\\
  \hat{\Pi}^{(13)}_1=\left(\begin{array}{ccc}
1/2&0&1/2\\
0&0&0\\
1/2&0&1/2\end{array}\right),\quad \hat{\Pi}^{(13)}_2=\left(\begin{array}{ccc}
1/2&0&-i/2\\
0&0&0\\
i/2&0&1/2\end{array}\right),\quad \hat{\Pi}^{(13)}_3=\hat{\Pi}^{(12)}_3,\\
  \hat{\Pi}^{(23)}_1=\left(\begin{array}{ccc}
0&0&0\\
0&1/2&1/2\\
0&1/2&1/2\end{array}\right),\quad \hat{\Pi}^{(23)}_2=\left(\begin{array}{ccc}
0&0&0\\
0&1/2&-i/2\\
0&i/2&1/2\end{array}\right),\quad \hat{\Pi}^{(23)}_3=\left(\begin{array}{ccc} 0&0&0\\
0&1&0\\
0&0&0\end{array}\right).
\end{eqnarray*}
To these we associate the dichotomic probabilities $(p_a^{(jk)}, 1-p_a^{(jk)})$ given by
\be \label{prho}
p_a^{(jk)}= \Tr\rho \hat{\Pi}_a^{(jk)}=\langle f^{(jk)}_a|\rho|f^{(jk)}_a\rangle.
\ee
These are indeed real positive numbers not greater than one, because $\rho$ is a positive Hermitean matrix (a state), whereas $
 \hat{\Pi}_a^{(jk)}$ are rank-one projectors (pure states).

 Notice that these dichotomic probabilities refer to $\hat{\Pi}^+$ projectors. Only for the qubit case the second component of the probability vector, namely $1-p$, can be obtained by projecting the density state on $\hat{\Pi}^-$. In general we have to choose either  positive or negative projectors to work. In this paper we use the positive ones.

 In order to fully determine the state $\rho$ we have to invert \eqn{prho} for the matrix elements of $\rho$.
As for the diagonal entries, we observe that  for the diagonal projectors it holds
\be\label{Pi3}
\hat{\Pi}_3^{(jk)} = E_{jj},\;\;\;  j<k
\ee
namely, they are given by the diagonal elements of the Weyl basis, hence independent on the  index $k$,  in the $(jk)$ plane, as we can verify in the table above, where $\hat{\Pi}^{(13)}_3=\hat{\Pi}^{(12)}_3= E_{11}$. This implies that the probabilities $p_3^{(jk)}, j<k\le 3$ only depend on the first of the two indices, $(jk)$, labelling the plane. We shall therefore use the notation
$
p_3^{(jk)}\rightarrow p_3^{(jj)}$, and,  we derive, from \eqn{prho}, \eqn{Pi3}
\beqa
\rho_{jj}&=& p_3^{(jj)}; \;\;\;  j=1,2\\
\rho_{33}&=&1-\sum_j^{2} p_3^{(jj)}
\eeqa
where $\Tr\rho=1$ has been used.

As for the off-diagonal entries of the matrix $\rho$, according to  Eq. \eqn{prho} we have to
 consider dichotomic probabilities associated to the off-diagonal projectors $\hat{\Pi}^{(jk)}_a, \; a=1,2$. These  allow to determine  the $6$ off-diagonal entries   $\rho_{jk}$ by means of  the relation
\be
\rho_{jk}= \Tr \rho^{(jk)}E_{jk}^T ;\;\;\; j<k
\ee
where we have introduced {\it auxiliary qubit states } $\rho^{(jk)}$ as follows
\be
 \rho^{(jk)}=S^{(jk)}_0+\left[2 p^{(jk)}_1-(p_3^{(jj)}+p_3^{(kk)})\right]S^{(jk)}_1-i \left[2 p^{(jk)}_2-(p_3^{(jj)}+p_3^{(kk)})\right]S^{(jk)}_2
%&+&\left[2 p^{(jk)}_3-(p^{(jj)}+p^{(kk)})\right]S^{(jk)}_3.
\ee
Matrix elements $\rho_{jk}, j>k$ are obtained by complex conjugation.
%\be
%\rho= \sum_{j<k} p_a^{(jk)} \hat{\Pi}^{(jk)}_a.
%\ee
Explicitly in terms of the Weyl basis we have then
\small{
\beqa\label{rhomat}
\rho&=& E_{11}p_3^{(11)} + E_{22}p_3^{(22)} + E_{33}(1- p_3^{(11)}-p_3^{(22)})\nn \\
 &+& \left(
E_{12} \left[p_1^{(12)}-\frac{1}{2}(p_3^{(11)}+p_3^{(22)})-i (p_2^{(12)}-\frac{1}{2}(p_3^{(11)}+p_3^{(22)}))\right] \right.\nn\\
&+& E_{13} \left[p_1^{(13)}-\frac{1}{2}(p_3^{(11)}+p_3^{(33)})-i (p_2^{(13)}-\frac{1}{2}(p_3^{(11)}+p_3^{(33)}))\right]\nn\\
&+& E_{23} \left[p_1^{(23)}-\frac{1}{2}(p_3^{(22)}+p_3^{(33)})-i (p_2^{(23)}-\frac{1}{2}(p_3^{(22)}+p_3^{(33)}))\right]  \nn\\
 &+& \left. h.c.\right).
\eeqa}
%namely \textcolor{red}{\bf The following has to be corrected accordingly}
%\small{
%\begin{eqnarray}\label{20}
% \rho=\left(\begin{array}{ccc}
%p^{(11)}&(p^{(12)}_{1}-1/2)-i(p^{(12)}_2-1/2)&(p^{(13)}_1-1/2)-i(p^{(13)}_2-1/2)\\
%(p^{(12)}_1-1/2)+i(p^{(12)}_{2}-1/2)&p^{(22)}&(p^{(23)}-{1}-1/2)-i(p^{(23)}_{2}-1/2)\\
%(p^{(13)}_{1}-1/2)+i(p^{)13)}_{2}-1/2)&(p^{(23)}_{1}-1/2)+i(p^{(23)}-{2}-1/2)&1-p^{(11)}-p^{(22)}
%\end{array}\right).\nonumber\\
%\end{eqnarray}}
Similarly to   the two-level system, the diagonal elements are associated to rank-one projectors  of positive eigenvalue of the observable $S_3^{(jk)}$, except for
$\rho_{33}$ which is obtained by the others through the constraint $\Tr\rho=1$.

\section{Qudit generalization}\label{qNit}
For a system with $d$ levels the Hilbert space $\mathcal{H}=\C^d$ is spanned by $d$ orthonormal vectors $ |  e_1\rangle, ...,  |  e_d\rangle$. The Lie algebra $\mathfrak{u}(d)$ can be described in terms of the Weyl basis $E_{jk}, j,k=1,...,d$.  As previously, we have   ${d!}/{2}$ different  $\mathfrak{u}(2)$ subalgebras, labelled by $(jk)$, $j\ne k$, and associate with  each of them  the Hermitean generators
\begin{eqnarray}
S_0^{(jk)}=\frac{1}{2}(E_{jj}+E_{kk}) \quad  S_3^{(jk)}= \frac{1}{2}(E_{jj}-E_{kk})\\
S_1^{(jk)}=\frac{1}{2}(E_{jk}+E_{kj}) \quad  S_2^{(jk)}=- \frac{i}{2}(E_{jk}-E_{kj})
\end{eqnarray}
acting on the $(j,k)$ plane. The eigenvalues of the operators $S_a^{(jk)}$
are equal to $+1/2$ and $-1/2$. They can be interpreted as spin
projections along the $x,y,z$ axes. For $d$-level atoms, these
eigenvalues and the corresponding eigenvectors may be  related to $j$th and $k$th
levels  when other levels are not excited.

Hence we construct the  rank-one projectors relative to the positive eigenvalue of each $S_a^{(jk)}, a=1,..,3$ generator according to
\be\label{projN}
\hat{\Pi}^{(jk)}_a= |f^{(jk)}_a\rangle\langle f^{(jk)}_a|= S_0^{(jk)} + S_a^{(jk)},\;\;\; a= 1,..,3.
\ee

Out of the ${d!}/{2}$ different $\mathfrak{u}(2)$ subalgebras, we select  those labelled by $(jk)$, with $j<k$. They are $d(d-1)/2$.
Hence we compute
\be\label{dicoprob}
p^{(jk)}_a= \Tr \rho \hat{\Pi}^{(jk)}_a
\ee
obtaining explicitly
\beqa
p_1^{(jk)}&=&\frac{1}{2}(\rho_{jj}+\rho_{kk})+ {\rm Re} \rho_{jk}\label{pr1}\\
p_2^{(jk)}&=&\frac{1}{2}(\rho_{jj}+\rho_{kk})- {\rm Im} \rho_{jk}\label{pr2}\\
p_3^{(jk)}&=&\rho_{jj}.\label{pr3}
\eeqa
They define  dichotomic probability vectors $(p^{(jk)}_a, 1-p^{(jk)}_a)$ for each $a=1,..,3$, and each couple $(j,k), j<k$, corresponding to $3d(d-1)/2$ probabilities.
Then we observe, as in the previous, two- and three-dimensional cases, that the diagonal projectors $ \hat{\Pi}^{(jk)}_3$ are independent of the second index in any $(jk)$ plane and coincide with the diagonal elements of the Weyl basis:
\be\label{diapro}
 \hat{\Pi}^{(jk)}_3= \hat{\Pi}^{(jh)}_3= E_{jj}, \;\;\; k\ne h, \;\; j=1,...,N-1\,.
 \ee
 This implies that, out of the $d(d-1)/2$  probabilities $p_3^{(jk)}$,   only $d-1$ are different. Thus the total number of independent parameters is $2 d(d-1)/2+ d-1= (d+1)(d-1)$.
 In other words, our choice of the $\mathfrak{u}(2)^{(jk)}$ subalgebras with $j<k$ provides us with a quorum.

Summarizing,  we are ready to state the following:

 \begin{theorem}\label{theodicho} Let $\rho$ be a qudit state and $p_a^{(jk)}, a=1,2,3, j<k = 1,...,N$  dichotomic probabilities, defined by Eq. \eqn{dicoprob}.
\begin{enumerate}[label=(\roman*)]
\item We have, for   the diagonal elements
 \be %\label{rhoelms}
 \rho_{jj}= \Tr\rho E_{jj}=p_3^{(jj)}, \;\;\; j=1,...,N-1; \;\;\; \rho_{N}= 1-\sum_j^{N-1} p_3^{(jj)} \label{diagN}
 \ee
 where we have re-labeled as previously $p^{(jk)}_3\rightarrow p_3^{(jj)}$.
\item The off-diagonal elements are obtained by
\be\label{offdiagN}
\rho_{jk}= \Tr \rho^{(jk)}E_{jk}^T , \;\;\; j<k
\ee
with
  \be\label{auxqb}
  \rho^{(jk)}= S^{(jk)}_0+[2 p^{(jk)}_1-(p_3^{(jj)} + p_3^{(kk)})]S^{(jk)}_1-i[2 p^{(jk)}_2-(p_3^{(jj)} + p_3^{(kk)})]S^{(jk)}_2
%+(2 p^{(jk)}_3-1)S^{(jk)}_3,
\;\;\;  1\le j<k\le N
 \ee
 auxiliary qubit states. The matrix elements $\rho_{kj} $ are  given  by  complex conjugation.
 \end{enumerate}
   \end{theorem}
\begin{proof} The  first statement is an immediate  consequence of Eq. \eqn{diapro}.  The second statement can be checked  by direct computation of the RHS of Eq. \eqn{offdiagN}, on using the auxiliary qubits \eqn{auxqb} and Eqs. \eqn{projN}, \eqn{dicoprob}.
\end{proof}
In terms of the  matrix elements obtained above in  Eqs. \eqn{diagN}, \eqn{offdiagN},  the explicit form of the density state $\rho$  in the Weyl basis can be readily written down and the result is a straightforward generalisation of Eq. \eqn{rhomat}.

Let us now discuss in full generality the dependence of the dichotomic probability representation on the choice of rank-one projectors. In the present case, for any $(j,k)$-plane we have
\be\label{PiN}
\hat{\Pi}^{(jk)}_a= S_0^{(jk)} + \vec x_a\cdot \vec S^{(jk)},
\ee
with $\vec S^{(jk)}=(S_1^{(jk)},S_2^{(jk)},S_3^{(jk)})$ and $(\vec x_a)_b= \delta_{ab}, \; a,b=1,2,3$.

As before, a  rotation of the three vectors $\vec x_a$ entails   rotated projectors $(\hat{\Pi}^{(jk)}_a)'$.
We have indeed
\be\label{xrotN}
  \vec x'_a= R^{(jk)} \vec x_a \longrightarrow  (\vec x_a\cdot \vec S^{(jk)})'= U^{(jk)}(\vec x_a\cdot \vec S^{(jk)}) {U^{(jk)}}^\dag, \;\;\;\;\; R^{(jk)}\in SO(3), U^{(jk)} \in SU(2)^{(jk)}
\ee
so that
\be\label{PirotN}
(\hat{\Pi}^{(jk)}_{a})' = U^{(jk)} \hat{\Pi}^{(jk)}_a {U^{(jk)}}^\dag
\ee
and
\be\label{rotlawN}
(p^{(jk)}_a)'(\rho)=\Tr \rho (\hat{\Pi}^{(jk)}_{a})' = \Tr {U^{(jk)}}^\dag \rho U^{(jk)} \hat{\Pi}^{(jk)}_{a}
\ee
yielding the transformation law  of the of dichotomic probabilities under rotation of rank-one projectors. Notice that, in order to preserve Eqs. \eqn{diapro} and \eqn{rhoelms} one has to choose one and the same rotation, $R^{(jk)}=R$,   in any $(j,k)$-plane.

\section{Reduction of the density matrix}\label{redumats}
We have shown in  previous section that  qudit states can be described in terms of
a set of different $(d^2-1)$ dichotomic probabilities
$\left(p_{1,2,3}^{(jk)},~1-p_{1,2,3}^{(jk)}\right)$ of classical-like random
variables. These probability distributions must satisfy the Silvester criterion
of  nonnegativity of the density operator, $\rho\geq 0$, i.e.,
eigenvalues of this operator must be nonnegative. Moreover the principal minors of the
operator $\rho$ in an arbitrary orthogonal basis must be nonnegative.

In this section we shall illustrate how these inequalities give rise to  quantum correlations for the  auxiliary qubits associated to qudit states.

To be definite, let us start with  a qudit state, $\rho$, represented by a $d\times d$ matrix, $d=n\cdot m$.
Let us consider two orthonormal bases,  $\;\{ | e_j\rangle,  j=1,...,n\}$, $\; \{|f_j\rangle,  j=1,...,m\}\, ,\; $ for the complex vector spaces $\C^n, \C^m$ respectively, and let us introduce in the space of $n\times n$, respectively   $m\times m $ complex matrices, the natural bases
\be
E_{jk}= |e_j\rangle \langle e_k| , j,k=1,...,n\;\;\;\; F_{jk}=| f_j\rangle \langle f_k|, j,k=1,...,m\,.
\ee
Hence, $\rho$ may be rewritten as follows
\be\label{rhotens1}
\rho= 
E_{11}\otimes R_{11} + E_{12}\otimes R_{12}+ ...+ E_{nn}\otimes R_{nn}
%= \sum_{j,k=1}^n E_{jk}\otimes R_{jk}
 \ee
with $R_{jk}$ $m\times m $ complex matrices defined by
\be
R_{jk}= \sum _{p,q=1}^m  \langle e_j \otimes  f_p| \rho |e_k \otimes f_q\rangle F_{pq}
\ee
so to have $\rho$ rearranged into $n^2$ blocks, each one  of $m\times m$ dimension
\be
\rho= \left(\begin{array}{cccc}R_{11}&R_{12} & ..& R_{1n}\\
..&..&..&.. \\
R_{n1}&R_{n2}&..&R_{nn}\end{array}\right)
\ee
We then define a $n\times n$ matrix $\rho_1$ by taking the partial trace over the second element of the tensor product
\be\label{rho1}
\rho_1:=\Tr_2 \rho= \sum_{j,k=1}^n E_{jk} \Tr R_{jk}
\ee
Alternatively, we can trace over the first element of the tensor product. Since $\Tr E_{jk} = \delta_{jk}$,  we
obtain a $m\times m$ matrix, $\rho_2$
\be\label{rho2}
\rho_2:=\Tr_1 \rho= \sum_{j,k=1}^n \delta_{jk} R_{jk}= R_{11}+...+R_{nn}.
\ee
We can actually exchange the role of the two bases, $E_{jk}, F_{jk}$ and express $\rho$ as follows
\be\label{rhotens2}
\rho
%= \tilde R_{11}\otimes F_{11} + \tilde R_{12}\otimes F_{12}+ ...+ \tilde R_{mm}\otimes F_{mm}
= \sum_{p,q=1}^m
 \tilde R_{pq}\otimes F_{pq}
 \ee
with $\tilde R_{pq}$ $n\times n $ complex matrices defined by
\be
\tilde R_{pq}= \sum _{j,k=1}^n  \langle f_p \otimes  e_j|\, \rho\, |f_q \otimes e_k\rangle E_{jk}
\ee
so to have $\rho$ rearranged into $m^2$ blocks of $n\times n$ dimension
\be
\rho= \left(\begin{array}{cccc}\tilde R_{11}&\tilde R_{12} & ..& \tilde R_{1m}\\
..&..&..&.. \\
\tilde R_{m1}&\tilde R_{m2}&..&\tilde R_{mm}\end{array}\right)
\ee
We then define a $n\times n$ matrix, $\tilde \rho_1$, by taking the partial trace over the second element of the tensor product in Eq. \eqn{rhotens2}
\be\label{trho1}
\tilde\rho_1= \Tr_2 \rho= \sum_{p,q=1}^m \delta_{pq} \tilde R_{pq}= \tilde R_{11}+...+\tilde R_{mm}.
\ee
By tracing over the first element we get instead
\be\label{trho2}
\tilde\rho_2= \Tr_1 \rho= \sum_{p,q=1}^m ( \Tr \tilde R_{pq}) F_{pq}.
\ee
Before showing that $\rho_{1,2}, \tilde \rho_{1,2}$ are all density states, namely nonnegative, Hermitean, trace-one complex matrices, for any value of $n,m$, let us see how the construction works for the simple
case of qu-quart states with $n=m=2$,
\be
\rho= \left(\begin{array}{cccc}
\rho_{11}&\rho_{12} &\rho_{13}& \rho_{14}\\
\rho_{21}&\rho_{22} &\rho_{23}& \rho_{24}\\
\rho_{31}&\rho_{32} &\rho_{33}& \rho_{34}\\
\rho_{41}&\rho_{42} &\rho_{43}& \rho_{44}
\end{array}\right) = \left(\begin{array}{cc}R_{11}&R_{12}\\R_{21}&R_{22}\end{array}\right),
\ee
with $R_{11}= R_{11}^\dag, R_{22}= R_{22}^\dag, R_{21}= R_{12}^\dag,$ $2\times 2$-block matrices.
The general expressions given above reduce therefore to
\be
\rho_1=\tilde\rho_2= \left(\begin{array}{cc}\Tr R_{11}&\Tr R_{12} \\ \Tr R_{21}&\Tr R_{22}\end{array}\right) \;\;\; {\rm and}\;\;\;  \rho_2= \tilde\rho_1= R_{11}+R_{22}.
\ee
It is readily seen that   the latter are Hermitian and trace-one matrices
\be
\rho_1^\dagger=\rho_1,\quad \mbox{Tr}\,\rho_1=1; \;\;\;\;\;\;\;\;\;\;
\rho^\dagger_2=\rho_2,\quad \mbox{Tr}\,\rho_2=1
.\label{D5}
\ee
Nonnegativity is proven below,  directly for the general case of a qudit, with  $d=nm$.

To this aim,  we shall need the following
well known results (see for example \cite{Bhatia, ReedSimon, Naimark}:
\begin{proposition}\label{prop1} Let $\mathcal{B}(\mathcal{H})$ denote the bounded operators on a Hilbert space $\mathcal{H}$. For any positive operator $ B\in \mathcal{B}(\mathcal{H})$ there exists an operator  $A\in \mathcal{B}(\mathcal{H})$ such that
$$
B= A^\dag A .
$$
\end{proposition}
\begin{proposition}\label{prop2} A density state, $\rho$,  is a positive linear functional over   $\mathcal{B}(H)$ iff it is nonnegative when evaluated on the positive elements of  $\mathcal{B}(H)$, that is,
$$
\rho(A^\dag A) \ge 0
$$
with $A\in \mathcal{B}(H)$.
\end{proposition}
The we can state the following
\begin{theorem}\label{thm4.1}
\; For a given $d$$\times$$d$ nonnegative trace-one Hermitian matrix,
with  $d=n m$, the reduced  matrices $\rho_{1,2}$, $\tilde \rho_{1,2}$, defined in Eqs. \eqn{rho1}, \eqn{rho2}, \eqn{trho1}, \eqn{trho2},
are trace-one Hermitean nonnegative matrices, i.e. they are quantum states.
\end{theorem}
\begin{proof}
Hermiticity and trace-one property are an immediate consequence of $\rho= \rho^\dag, \Tr\rho=1$.

 %Nonnegativity of $\rho_2, \tilde\rho_1$, is proven as a straightforward generalization of the ququart case. To this, let us consider the $N$-dimensional vectors  $\;  | \Psi_1\rangle=(\vec \psi_1,\vec 0, ..,\vec 0), \;  | \Psi_2\rangle=(\vec 0,\vec \psi_2, \vec 0, .., \vec 0),...,  | \Psi_n\rangle=(\vec 0,..., \vec 0,\vec \psi_n)$ where $\vec \psi_j, j=1,..,n $ as well as $\vec 0$ are $m$-dimensional vectors. Since $\rho$ is nonnegative we have
%$\;
%\langle\Psi|\,\rho\,|\Psi\rangle\ge 0 \;
%$ , for any vector $|\Psi\rangle$,
%which implies,
%\be
%0\le \langle\Psi _j|\,\rho\,|\Psi_j\rangle= (\vec\psi_j, R_{jj}\vec\psi_j),
%\ee
%from which it follows
%\be
%(\vec\psi, \rho_2\vec\psi)= (\vec\psi, \sum_j R_{jj}\vec\psi)  \ge 0
%\ee
%for $\vec \psi$ a generic m-dimensional vector. The proof for $\tilde \rho_1$ is identical except for the fact that the latter is now $n$-dimensional; hence, vectors $\vec \psi_j, j=1,..,m$ are $n$-dimensional.

In order to prove nonnegativity of $\rho_1$ we advocate the two propositions quoted above.
Let us take $\rho$ in the form \eqn{rhotens1} and evaluate it over the positive operator $(A \otimes {\bf 1})^\dag (A\otimes {\bf 1})= A^\dag A \otimes {\bf 1}$. According to Prop. \ref{prop2} we have
$$
0\le \rho(A^\dag A \otimes {\bf 1})=\sum_{j,k} E_{jk}(A^\dag A) \Tr R_{jk}= \rho_1(A^\dag A)
$$
that is, according to Prop. \ref{prop1}, $\rho_1$ is nonnegative. Nonnegativity of $\tilde \rho_2$ can be proven in the same way, by representing $\rho$ in the form \eqn{rhotens2}.

Analogously,  to prove nonnegativity of $\rho_2$ we  take again $\rho$ in the  form \eqn{rhotens1} but evaluate it over the positive operator $({\bf 1} \otimes { A})^\dag ({\bf 1} \otimes { A})={\bf 1} \otimes  A^\dag A $. According to Prop. \ref{prop2} we have
$$
0\le \rho({\bf 1}\otimes A^\dag A)=\sum_{j,k} E_{jk}({\bf 1}) R_{jk} (A^\dag A)= \rho_2(A^\dag A)
$$
that is, according to Prop. \ref{prop1}, $\rho_2$ is nonnegative. Nonnegativity of $\tilde \rho_1$ can be proven in the same way, by representing $\rho$ in the form \eqn{rhotens2}.
\end{proof}
\subsection{Polynomial roots of probabilities}
As a direct consequence of Theorem \ref{thm4.1}  we may  derive new interesting inequalities.
To this, let us consider the characteristic polynomial  in $\lambda$, associated to $\rho$,  $d\times d$, Hermitean, positive, trace-one matrix,  which may be written as
$$
\det\left(\rho-\lambda{\bf 1}\right)=\sum_{k=1}^d c_k\lambda^k=\prod_{k=1}^d(\lambda-\lambda_k),
$$
where $\lambda_k\geq 0 , k=1,...,d$ are the eigenvalues of $\rho$,
%$\lambda_1,\lambda_2,\ldots,\lambda_N\geq 0$,
and $\sum\lambda_k=1$. Then, the solution of the eigenvalues   equation
\begin{equation}\label{6}
\det\left(\rho-\lambda{\bf 1}\right)=0
\end{equation}
yields a probability vector $(\lambda_1,...,\lambda_d)$. By virtrue of Theorem \ref{thm4.1}, the following Corollary holds
\begin{corollary} Let $d=nm$, $\rho$ a qudit and
 $\rho_1$, $\rho_2$ respectively n-and m-dimensional states defined in   Eq. \eqn{rho1}, \eqn{rho2}. Let us consider the  associated characteristic polynomials
\begin{eqnarray}
\det\,\big(\rho_1-\lambda {\bf 1}_{n\times
n}\big)=\prod_{s_1=1}^n(\lambda-\Lambda_{s_1}),\label{7}\\
\det\,\big(\rho_2-\lambda {\bf 1}_{m\times
m}\big)=\prod_{s_2=1}^m(\lambda-\bar\Lambda_{s_2}).\label{8}
\end{eqnarray}
Then,
\be
 0\leq\Lambda_{s_1}\leq 1\;\;\;  0\leq\bar\Lambda_{s_2}\leq 1.
 \ee
Moreover, the map  which associates to  the probability distribution
$\lambda_1,\lambda_2,\ldots,\lambda_d$  the probability
distributions $ \Lambda_1,\Lambda_2,\ldots,\Lambda_n$ and
$\bar\Lambda_1,\bar\Lambda_2,\ldots,\bar\Lambda_m$, is bijective.
\end{corollary}
By explicitly computing the LHS of Eqs. \eqn{7},\eqn{8} one can derive bounds on the determinant of the states $\rho_1$, $\rho_2$.  Similar inequalities can be obtained starting with  $\tilde\rho_1, \tilde\rho_2$, defined in Eqs. \eqn{trho1}, \eqn{trho2}.

As an
example, let us consider the case  of $d=4$. With a straightforward calculation we find
\begin{eqnarray}
\Lambda_{1,2}=\frac12\left[1\pm\sqrt{1-4\det\,\rho_1}\right] \Longrightarrow \quad 0\leq\det\,\rho_1\leq 1/4,\nonumber\\
\label{9}\\
\bar\Lambda_{1,2}=\frac12\left[1\pm\sqrt{1-4\det\,\rho_2}\right]\Longrightarrow\quad 0\leq\det\,\rho_2\leq
1/4.\nonumber
\end{eqnarray}
The inequality for the  determinant of the state $\rho_1$ can be easily checked to be true for the general case   $d=2\,m$ ($\rho_1$ being two-dimensional again). We have then
%$\rho=\left(\begin{array}{cc}R_1&R_2\\R_3&R_4\end{array}\right)$,
%$R_3=R_2^\dagger $. The corresponding matrices
%$\rho(1)=\left(\begin{array}{cc}\mbox{Tr}\,R_1&\mbox{Tr}\,R_2\\
%\mbox{Tr}\,R_3&\mbox{Tr}\,R_4\end{array}\right)$ satisfy the inequality
\begin{equation}\label{10}
0\leq\left(\mbox{Tr}\,R_{11}\right)\left(\mbox{Tr}\,R_{22}\right)-
\left(\mbox{Tr}\,R_{12}\right)\left(\mbox{Tr}\,R_{21}\right)\leq 1/4,
\end{equation}
or
\begin{equation}\label{11}
  \left(\mbox{Tr}\,R_{11}\right)\left(\mbox{Tr}\,R_{22}\right)\geq
\left(\mbox{Tr}\,R_{12}\right)\left(\mbox{Tr}\,R_{21}\right),\quad
\left(\mbox{Tr}\,R_{12}\right)\left(\mbox{Tr}\,R_{21}\right)+1/4\geq
\left(\mbox{Tr}\,R_{11}\right)\left(\mbox{Tr}\,R_{22}\right).
\end{equation}

The new inequalities are susceptible to be checked experimentally for  density
matrices obtained within the framework of quantum tomography
approach.

\section{New information-entropic inequalities for nonnegative
trace-one Hermitian matrices}\label{infoentro}

As an application of the results of previous sections, we derive in this  section new information-entropic inequalities for density states. By using  Eqs. \eqn{pr1}-\eqn{pr3}, we can express the dichotomic probabilities $(p_a^{(jk)}, 1- p_a^{(jk)}), a= 1,2,3$ in terms of the matrix elements of  $\rho$. Upon substituting them in expressions like von Neumann or Tsallis relative entropy we get the desired inequalities as follows.

For dichotomic variables relative von Neumann entropy reads
\be
S_{vN}= -[p\ln p + (1-p) \ln(1-p)]\ge 0
\ee
so that
\be\label{inevN}
-[p_a^{(jk)}\ln\frac{p_a^{(jk)}}{1-p_a^{(jk)}}+ \ln (1- p_a^{(jk)}) ]\ge 0%\label{vN}
\ee
Analogously, for Tsallis relative entropy we have
\be\label{tsallis}
S_{Ts}= (1-q)^{-1}\{(p_a^{(jk)})^q(p_b^{(jk)})^{1-q}+ (1-p_a^{(jk)})^q(1-p_b^{(jk)})^{1-q}-1\}\ge 0.
\ee
In the particular qubit case, we can drop the $(jk)$ index, and get
\be\label{2pr}
\begin{array}{ccccccc}
p_1&=&\frac{1}{2}  + {\rm Re} \rho_{12} &\;\;\; & 1- p_1&=&\frac{1}{2} - {\rm Re} \rho_{12}\\
p_2&=&\frac{1}{2}  - {\rm Im} \rho_{12}&\;\;\; & 1- p_2&=&\frac{1}{2} + {\rm Im} \rho_{12}\\
p_3&=&\rho_{11} &\;\;\; & 1- p_3&=&1-\rho_{11}
\end{array}
\ee
and Eq. \eqn{inevN} becomes
\beqa
&& \ln \sqrt{\frac{1}{4}-({\rm Re}\, \rho_{12})^2} + {\rm Re}\, \rho_{12}\ln \frac{\frac{1}{2} +{\rm Re}\, \rho_{12}}{\frac{1}{2} -{\rm Re}\, \rho_{12}}\le 0 \label{inevN1}\\
&&\ln \sqrt{\frac{1}{4}-({\rm Im}\, \rho_{12})^2} + {\rm Im}\, \rho_{12}\ln \frac{\frac{1}{2} +{\rm Re}\, \rho_{12}}{\frac{1}{2} -{\rm Re}\, \rho_{12}}\le 0\label{inevN2}
\eeqa
Assuming $a=1, b=2$ Eq. \eqn{tsallis} for the Tsallis relative entropy becomes in turn,
\be
(1-q)^{-1}\{(\frac{1}{2}  + {\rm Re} \rho_{12} )^q (\frac{1}{2}  - {\rm Im} \rho_{12} )^{1-q}+ (\frac{1}{2}  - {\rm Re} \rho_{12} )^q(\frac{1}{2}  + {\rm Im} \rho_{12} )^{1-q}-1\}\ge 0. \label{inets}
\ee
Finally, in the qudit case, with $N= 2n$, we can apply our findings to the $2\times 2$ state
\be
\rho_1= \left( \begin{array}{cc} \Tr R_{11} & \Tr R_{12}\\
\Tr R_{21}& \Tr R_{22}
\end{array}\right),
\ee
and obtain, from  inequalities \eqn{inevN1},\eqn{inevN2}, \eqn{inets}, new inequalities by substituting $\rho_{12}$ with $\Tr R_{12}$.
%\begin{equation}\label{13}
% 
%\frac{\left[\mbox{Re}\left(\mbox{Tr}\,R_1\right)\right]^q\left[\mbox{Im}\left(\mbox{Tr}\,R_1\right)\right]^{1-q}+
%\left[1-\mbox{Re}\left(\mbox{Tr}\,R_1\right)\right]^q\left[1-\mbox{Im}\left(\mbox{Tr}\,R_1\right)\right]^{1-q}}{1-q}\geq
%0.
%\end{equation}
%Also
% 
%\begin{equation}\label{14}
%-\left[\mbox{Re}\left(\mbox{Tr}\,R_1\right)\right]\ln
%\left[\mbox{Re}\left(\mbox{Tr}\,R_1\right)\right]-\left[\mbox{Im}\left(\mbox{Tr}\,R_2\right)\right]
%\ln\left[\mbox{Im}\left(\mbox{Tr}\,R_2\right)\right]\geq 0.
%\end{equation}
%The other inequality reads
%\begin{eqnarray}\label{15}
%\left[\mbox{Re}\left(R_1+R_4\right)_{\alpha\beta}\right]\ln\left\{\frac
%{\mbox{Re}\left(R_1+R_4\right)_{\alpha\beta}}{\mbox{Im}\left(R_1+R_4\right)_{\alpha\beta}}
%\right\}\nonumber\\
%-\left[1-\mbox{Re}\left(R_1+R_4\right)_{\alpha\beta}\right]\ln\left\{\frac
%{1-\mbox{Re}\left(R_1+R_4\right)_{\alpha\beta}}{1-\mbox{Im}\left(R_1+R_4\right)_{\alpha\beta}}
%\right\}\geq 0.
%\end{eqnarray}
%These inequalities follow from known entropic inequalities for
%probability distributions of dichotomic random variables.

%\newpage

%\newpage

\section{Conclusions}\label{concl}

To conclude, we point out the main results of our study.

We proved that a $d$-dimensional density state,  $\rho$, has matrix elements which can be
parameterized in terms of  dichotomic
probability distributions and we discussed the dependence of such a representation on the chosen basis of rank-one projectors. The expression of matrix elements
$\rho_{jk}$ of the qudit in terms of dichotomic probabilities
is the argument of  Theorem \ref{theodicho}. The probabilities
$ p_a^{(jk)}$  satisfy the Silvester criterion
of nonnegativity of the density matrix $\rho$.
These rigorously proven expressions for the density matrix of qudit states in
terms of dichotomic probabilities $p_a^{(jk)}$ are the main result of
this study.

 It is worth noting that a possibility to reconstruct the matrix
elements of the density operator in discrete basis was suggested
in~\cite{PRL1995} without obtaining the dichotomic probability representation
of the density matrix; it was related to  experiments where photon-number
distributions were measured to determine the density matrices of photon
states.

Upon elaborating on previous
claims~\cite{Chernega-JRLR,Chernega-JPCS,RitaEntropy,PS-Milestones,Julio-Entropy,MA-Entropy}
we proved that it is possible to define reduced matrices from the
original qudit, where $d=nm$, and obtain smaller $n$$\times$$n$-
and $m$$\times$$m$-dimensional matrices, which keep the properties
of the initial matrix $\rho=\rho^\dagger$, Tr\,$\rho=1$, and
$\rho\geq 0$ of being states. The theorem can be extended
iteratively to matrices with $d=n_1n_2\cdots n_m$.

We obtained  new relations for the determinants
and eigenvalues of reduced states. We   derived new inequalities, including entropic inequalities for
the matrix elements of the qudit, which provide new
relations for its matrix elements. These inequalities can be employed to
control the accuracy of experiments where  density matrix elements are
reconstructed, in particular by using tomographic methods.

% out that it would be possible to check experimentally these inequalities by using tomographic methods.
%This result is obtained, by using the family of the $u(2)$ subalgebra generators obtained by  immersion of  Pauli matrices into the space of $N\times N$ -Hermitian positive matrices.
%The immersion may be described  in the following simple manner. One removes  all zero nondiagonal matrix elements
%from this $N \times  N$-matrix, the matrices become 2$\times$2-matrices
%coinciding with standard Pauli matrices describing the spin-1/2 projections on
%the $x,y,z$ directions and determining the probability representation of qubit
%states described by Eq.~(\ref{3}). This fact allows to
%rigorously obtain the universal expression for the $N\times N$ density
%matrix; here, the integer $N=2,3,\ldots,\infty$.

The description of quantum states in terms of dichotomic probabilities amounts to decompose a point in the equilateral triangle by using an orthogonal decomposition with respect to the edges instead of using the vertices.This decomposition may be used to describe the evolution on the space of quantum systems in terms of evolution in their dichotomic probability distributions. For example, for systems coupled to an
environment (open systems) the Markovian or non-Markovian evolution of qudit states studied
in \cite{Chrus} could be mapped onto the time evolution of the associated dichotomic probabilities. It is conceivable that one has to  consider higher order ordinary differential equations in order to describe the first order differential equations  GKLS on the space of quantum states. We
plan to perform such an analysis in a forthcoming publication.

\bigskip

\noindent{\bf Acknowledgements}

\noindent
G.M. would like to thank the support provided by the Santander/UC3M Excellence Chair Programme 2019/2020; he also  acknowledges financial support from the Spanish Ministry of Economy and Competitiveness, through the Severo Ochoa Programme for Centres of Excellence in RD (SEV-2015/0554). G.M. is a member of the Gruppo Nazionale di Fisica Matematica (INDAM), Italy.
P.V.  acknowledges  support by COST (European Cooperation in Science  and  Technology)  in  the  framework  of  COST  Action  MP1405  QSPACE.

%\section*{Bibliography}

\end{document}